\documentclass[a4paper,11pt]{article}

\usepackage{preamble-fv-simple}
\usepackage{macros}

\usepackage{needspace}

\author[1]{Lars Jaffke} 
\author[2]{Paloma T.\ Lima} 
\author[3]{Roohani Sharma} 

\affil[1]{University of Bergen, Bergen, Norway} 
\affil[ ]{\texttt{lars.jaffke@uib.no}} 
\affil[2]{IT University of Copenhagen, Copenhagen, Denmark} 
\affil[ ]{\texttt{palt@itu.dk}} 
\affil[3]{Max Planck Institute for Informatics, Saarland Informatics Campus, Saarbr{\"u}cken, Germany} 
\affil[ ]{\texttt{rsharma@mpi-inf.mpg.de}} 

\title{$b$-Coloring Parameterized by Pathwidth is XNLP-complete} 

\begin{document}

\maketitle

\begin{abstract}
    We show that the \bCol problem
	is complete for the class \XNLP when parameterized by the pathwidth 
    of the input graph.
	Besides determining the precise parameterized complexity of this problem,
	this implies that \bCol parameterized by pathwidth is $\W[t]$-hard for all~$t$,
	and resolves the
	parameterized complexity of \bCol parameterized by treewidth.
\end{abstract}

\section{Introduction and Definitions}
A \emph{$b$-coloring} of a graph $G$ is a proper vertex-coloring such that each color class has 
a vertex, called \emph{$b$-vertex}, that has a neighbor in each color class except its own.
This problem originated in the study of a heuristic for the \textsc{Graph Coloring} problem,
for more details see for instance~\cite{IrvingManlove99,JaffkeLL21,SilvaThesis}.
In this note we resolve the open question of the parameterized complexity of the \bCol 
problem parameterized by treewidth, 
which has been posed for instance in~\cite{JaffkeLL21,SilvaThesis}.
In particular, we show that already in the more restrictive parameterization by \emph{pathwidth},
the problem is \XNLP-complete, which implies it is $\W[t]$-hard for all $t$~\cite{XNLP-comp}.
To show hardness, we reduce from an orientation problem parameterized by pathwidth
that has recently been shown to be \XNLP-complete~\cite{BodlaenderCW22}.

\myparagraph{Basic notations and definitions.}
For two integers $a \le b$ we let $[a..b] = \{a, a+1, \ldots, b\}$,
and for a positive integer $a$, we let $[a] = [1..a]$.
All graphs considered here are finite and simple.
For an (undirected or directed) graph $G$, 
we denote its vertex set by $V(G)$ and its edge set by $E(G)$.
For an edge $\{u, v\} \in E(G)$, we use the shorthand ``$uv$''.
If $G$ is a directed graph, then denoting the edge $e = (u, v) \in E(G)$ by $uv$
also points to $e$ being directed from $u$ to $v$.
Given an undirected graph $G$, an \emph{orientation} of $G$ is a directed graph obtained 
from $G$ by replacing each edge $\{u, v\} \in E(G)$ by either $(u, v)$ or $(v, u)$.
A set $S \subseteq V(G)$ is \emph{independent} if for all pairs of distinct $u, v \in S$, 
$uv \notin E(G)$.
A \emph{star} is an undirected graph with one special vertex called the \emph{center}
that is adjacent to all of the remaining vertices, called \emph{leaves}, 
which form an independent set.

\fancyproblemdef
    {$b$-Coloring}
    {Undirected graph $G$, integer $k$}
    {Does $G$ have a $b$-coloring with $k$ colors?}

\begin{definition}
    Let $G$ be a graph.
    A \emph{path decomposition} of $G$ is a sequence $\calB$ = $B_1, \ldots, B_d$ 
    of subsets of $V(G)$ called \emph{bags} covering $V(G)$ such that:
    \begin{enumerate}
        \item For each edge $e \in E(G)$, there is some $i \in [d]$ such that $e \subseteq B_i$.
        \item For each $h, i, j \in [d]$ with $h < i < j$, $B_h \cap B_j \subseteq B_i$.
    \end{enumerate}
    The \emph{width} of $\calB$ is $\max_{i \in [d]} \card{B_i} - 1$,
    and the \emph{pathwidth} of $G$ is the smallest width of all its path decompositions.
\end{definition}

Membership in \XNLP will follow from the membership of \bCol parameterized by a linear width measure
with more expressive power than pathwidth,
namely a width measure equivalent to linear clique-width.
We define it next and show its relation to pathwidth.
\begin{definition}
    Let $G$ be a graph and $S \subseteq V(G)$.
    The \emph{module number} of $S$ is the number of equivalence classes of the equivalence relation $\sim_S$ defined as: $u \sim_S v \Leftrightarrow N(u) \cap (V(G) \setminus S) = N(v) \cap (V(G) \setminus S)$.
    Let $\pi = v_1, \ldots, v_n$ be a linear order of $V(G)$.
    The \emph{module-width} of $\pi$ is the maximum, over all $i$, of the module number of $\{v_1, \ldots, v_i\}$.
    The \emph{linear module-width} of $G$ is the minimum module-width over all its linear orders.
\end{definition}

\begin{lemma}\label{lem:pw:to:mw}
    Let $G$ be a graph and $\calB$ be a path decomposition of $G$ of width $w$.
    Then one can construct in polynomial time and logarithmic space a linear order of module-width at most $w + 2$.
\end{lemma}
\begin{proof}
    For each vertex $v$, let $B_v$ be the leftmost bag (i.e., the bag with the smallest index) 
    of $\calB$ containing $v$.
    Let $\pi = v_1, \ldots, v_n$ be a linear order of $V(G)$ such that the bags $B_{v_1}, \ldots, B_{v_n}$ appear in the same order as in $\calB$, with ties broken arbitrarily.
    Clearly, this order can be constructed within the claimed time and space bounds;
    we argue that it has module-width at most $w + 2$.
    By the properties of a path decomposition,
    for each $i \in [n]$, there are at most $w + 1$ vertices in $\{v_1, \ldots, v_i\}$ that have a neighbor in $\{v_{i+1}, \ldots, v_n\}$.
    Therefore the module-number of each such $\{v_1, \ldots, v_i\}$ can be at most $w+2$: $w+1$ for the aforementioned vertices and one for the vertices without neighbors in $\overline{V_i}$.
\end{proof}

\myparagraph{The class \XNLP.}
We assume familiarity with the basic technical notions of parameterized complexity and refer 
to~\cite{DowneyF99} for an overview.
The class \XNLP,
introduced as $N[f~\mathrm{poly}, f~\log]$ by Elberfeld et al.~\cite{ElberfeldST15},
consists of the parameterized decision problems that 
given an $n$-bit input with parameter $k$ can be
solved by a non-deterministic algorithm that simultaneously uses at most $f(k) n^c$ time and at most 
$f(k) \log n$ space,
where $f$ is a computable function and $c$ is a constant. 
We refer to~\cite{XNLP-comp,ElberfeldST15} for more details on this complexity class.
Hardness in \XNLP is transferred via \emph{parameterized logspace reductions}~\cite{ElberfeldST15}
which are parameterized reductions in the traditional sense~\cite{DowneyF99}
with the additional constraint of using only $f(k) + \calO(\log n)$ space,
where once again $k$ is the parameter of the problem and $n$ is the input size.

\section{The proof}
Adapting the \XP-algorithm for \bCol parameterized by module-width $w$~\cite{JaffkeLL21}
to a nondeterministic \FPT-time and $f(w)\log n$ space algorithm,
we can show that \bCol parameterized by linear module-width, and therefore by pathwidth, is in \XNLP.
This can be done similarly as in the case of \textsc{Graph Coloring} parameterized by linear clique-width as shown in~\cite{BodlaenderGJJL22}.
\begin{lemma}\label{lem:membership}
    \bCol parameterized by the module-width of a given linear order of the vertices of the input graph is in \XNLP.
\end{lemma}
We start from the following problem in our reduction
which is known to be \XNLP-complete
when parameterized by the width of a given path decomposition of the input graph~\cite{BodlaenderCW22}.
\fancyproblemdef
	{\CircOri}
	{Undirected graph $G$ with edge weights $\weight \colon E(G) \to \bN$ given in unary.}
	{Is there an orientation $\ori{G}$ of $G$ such that for each $v \in V(G)$: 
	$$\sum_{vx \in E(\ori{G})} \weight(vx) = \sum_{xv \in E(\ori{G})} \weight(xv)$$\vspace*{-.5cm}}

\begin{theorem}
	\bCol parameterized by the width of a given path decomposition of the input graph is \XNLP-complete.
\end{theorem}
\begin{proof}
    Membership follows from \cref{lem:pw:to:mw,lem:membership}.
    To show \XNLP-hardness,
    we give a parameterized logspace-reduction from 
	the \CircOri problem parameterized by the width of a given path decomposition of the input graph,
	which was 
	shown to be \XNLP-complete in~\cite{BodlaenderCW22}.
	Let $(G, \weight)$ be an instance of \CircOri, given with a path decomposition $\calB$ of $G$.
	We let $n = \card{V(G)}$, $m = \card{E(G)}$,
	and $\bfW = \sum_{e \in E(G)} \weight(e)$.
	For each vertex $v \in V(G)$, we let $W_v = \sum_{uv \in E(G)} \weight(uv)$.
	We may assume that $G$ is connected and 
	that for all $e \in E(G)$, $\weight(e) \ge 1$;
	therefore $\bfW \ge m \ge n-1$.
	
	We construct an equivalent instance $(H, k)$ of \bCol.
	We let 
	\begin{align}
		k = 2\bfW + 3m + n + 2.
	\end{align}
	We begin the construction of $H$ 
	which is illustrated in \cref{fig:sketch}
	by adding $2\bfW + 2$ disjoint copies of a star with $k-1$ leaves.
	Let $S^\star$ be one of these stars. We denote its center by $s^\star$ and 
	refer to it throughout the proof as the \emph{superstar}.
	The remaining ones are referred to as \emph{anonymous}.
	We partition a subset of the leaves of $S^\star$ into $\calL = \{L_{e,v} \mid e \in E(G), v \in e\}$
	where for all $e \in E(G)$ and $v \in e$, $\card{L_{e,v}} = \weight(e)$. 
	Note that this is possible since $k-1 \ge 2\bfW $.
	
	\myparagraph{Vertex gadget.}
	For each $v \in V(G)$,
	we add $v$, as well as 
	a set $P_v$ of $k - \frac{3}{2}W_v - 1$ independent vertices 
	to $H$.
	We add all edges between $v$ and $P_v$, 
	Furthermore, for each edge $e \in E(G)$ such that $v \in e$,
	we connect $v$ and the vertices in $L_{e,v}$ in $H$.
	
	\myparagraph{Edge gadget.}
	For each $e = uv \in E(G)$, 
	we add the following gadget to $H$.
	First, it has two vertices $x_{e,u}$ and $x_{e,v}$, 
	a set $Y_e$ of $\weight(e)$ vertices,
	and a set $Z_e$ of $k - 2\weight(e) - 3$ vertices.
	The vertex $x_{e,u}$ is adjacent to $Y_e \cup Z_e \cup L_{e,u}$, 
	and $x_{e,v}$ is adjacent to $Y_e \cup Z_e \cup L_{e,v}$.
	We make $u$ and $v$ adjacent to~$Y_e$.
	We furthermore add two new vertices $q_{e, 1}$ and $q_{e, 2}$ to $H$ that are connected by an edge,
	as well as all edges between $q_{e, h}$ and $Z_e \cup L_{e, u} \cup L_{e, v} \cup \{x_{e, u}, x_{e, v}\}$
	for all $h \in [2]$.
	We let $X = \{x_{e, u}, x_{e, v} \mid e = uv \in E(G)\}$,
	and $\calQ = \{q_{e, 1}, q_{e, 2} \mid e \in E(G)\}$. 
	
	\begin{figure}
	    \centering
	    \includegraphics[height=.3\textheight]{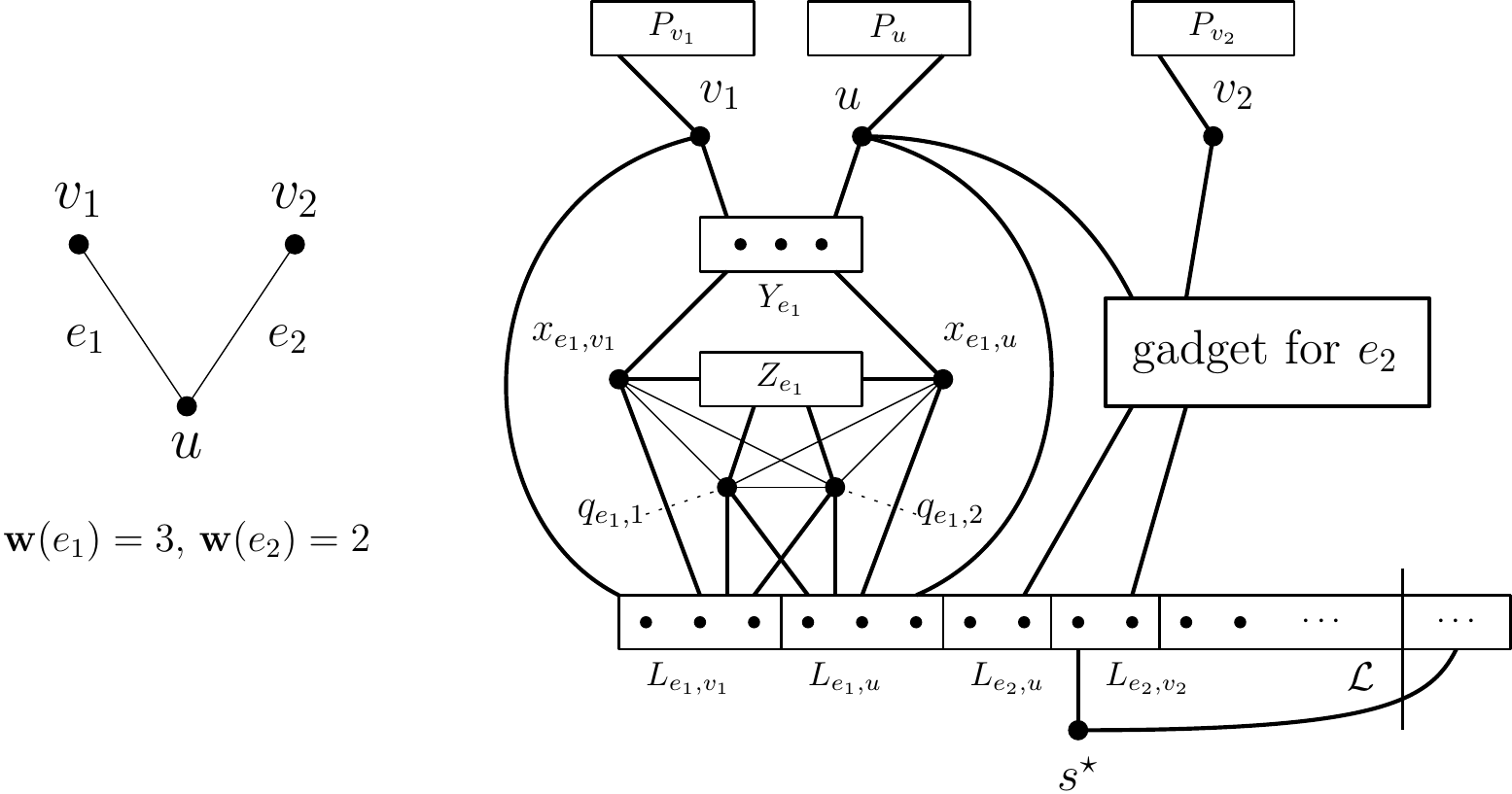}
	    \caption{Sketch of the main part of the reduction. 
	        Bold edges mean that all edges between the corresponding sets are present. All vertex sets represented by single boxes are independent.
	        Note that $\card{L_{e, v_1}} = \card{L_{e, u}} = \card{Y_{e_1}} = \weight(e_1) = 3$ and recall that $\card{Z_{e_1}} = k - 2\weight(e_1) - 3$.}
	    \label{fig:sketch}
	\end{figure}
	
	\medskip
	\noindent%
	Adding all vertex and edge gadgets finishes the construction of $H$,
	which can be performed using only logarithmic space. 
	
	\needspace{3\baselineskip}
	\begin{nestedclaim}
		If $(G, \weight)$ has a circulating orientation,
		then $H$ has a $b$-coloring with $k$ colors.
	\end{nestedclaim}
	\begin{claimproof}
		Let $\ori{G}$ be the circulating orientation of $(G, \weight)$.
		We give a coloring of the vertices of $H$ with colors $[0..(k-1)]$. 
		To do so, we identify some important subsets of $[0..(k-1)]$ whose $b$-vertices will appear
		in targeted regions of $H$.
		First, we let 
		$V(G) = \{v_1, \ldots, v_n\}$ and 
		$E(G) = \{e_1, \ldots, e_m\}$.
		We construct a proper coloring of $H$ such that once the coloring is completed,
		the following hold.
		\begin{enumerate}
			\item\label{enum:COtoCOL:1}
			The vertex $s^\star$ (the center of the superstar) is a $b$-vertex of color $0$.
			\item\label{enum:COtoCOL:2} 
			For each $i \in [n]$, $v_i$ is a $b$-vertex of color $i$.
			\item\label{enum:COtoCOL:3} 
			For each $i \in [m]$, $q_{e_i, 1}$ is a $b$-vertex of color $n+i$, 
			    and $q_{e_i, 2}$ is a $b$-vertex of color $m+n+i$.
			\item\label{enum:COtoCOL:4} 
			For each $i \in [m]$, either $x_{e_i, u}$ or $x_{e_i, v}$,
				where $e_i = uv$, is a $b$-vertex of color $2m+n+i$.
			\item\label{enum:COtoCOL:5} 
			Each of the remaining $k - (3m+n+1) = 2\bfW + 1$ colors has a $b$-vertex that is a center 
				of an anomymous star.
		\end{enumerate}
		
		Let $S_1, \ldots, S_{2\bfW+1}$ be the anonymous stars with centers $s_1, \ldots, s_{2\bfW+1}$, respectively.
		For each $i \in [2\bfW+1]$, we assign $s_i$ the color $3m+n+i$, 
		and the leaves of $S_i$ the colors $[0..(k-1)] \setminus \{3m+n+i\}$. This satisfies~\cref{enum:COtoCOL:5}.
		We assign $s^\star$ the color $0$
		and its leaves the colors $[k-1]$ in such a way that 
		colors $[(3m+n+1)..(3m+n+2\bfW)]$ appear on the vertices in $\calL$.
		This satisfies~\cref{enum:COtoCOL:1}.
		For each $i \in [n]$, we assign $v_i$ color $i$.
		For each $i \in [m]$ and each $v \in e_i$, we let $C_{e_i, v}$ be the colors appearing on $L_{e_i, v}$,
		and we assign $x_{e_i, v}$ the color $2m+n+i$.
		
		We now color the edge gadgets.
		Let $i \in [m]$ and $e_i = uv$.
		We give $q_{e_i, 1}$ color $n+i$ and $q_{e_i, 2}$ color $m+n+i$.
    	We assign the vertices in $Z_{e_i}$ the colors 
    	$$[0..(k-1)] \setminus (C_{e_i, u} \cup C_{e_i, v} \cup \{n+i, m+n+i, 2m+n+i\}).$$
		If $e_i$ is directed from $u$ to $v$ in $\ori{G}$,
		then we repeat colors $C_{e_i, u}$ on $Y_{e_i}$.
		Observe that this makes $x_{e_i, v}$ a $b$-vertex for color $2m+n+i$:
		it sees colors $C_{e_i, v}$ on $L_{e_i, v}$, colors $C_{e_i, u}$ on $Y_{e_i}$,
		and the remaining colors other than its own on $Y_{e_i} \cup \{q_{e_i, 1}, q_{e_i, 2}\}$.
		Moreover, $q_{e_i, 1}$ is a $b$-vertex for color $n+i$,
		since it sees color $m+n+i$ on $q_{e_2, i}$,
		color $2m+n+i$ on $x_{e_i, v}$,
		and the remaining colors on $L_{e, u} \cup L_{e, v} \cup Z_e$.
		Similarly, $q_{e_i, 2}$ is a $b$-vertex for color $m+n+i$.
		Once this is done for all $i$,~\cref{enum:COtoCOL:3,enum:COtoCOL:4} are satisfied.

		We now color the vertex gadgets.
		We first argue that each
		$v \in V(G)$ already sees precisely $\frac{3}{2}W_{v}$ colors in its neighborhood.
		This is because $v_i$ sees
		$W_{v}$ colors on $\bigcup_{e \in E(G), v \in e} L_{e, v}$,
		and for each edge $e$ that is directed towards $v$,
		there are $\weight(e)$ additional colors appearing in the neighborhood of $v$;
		concretely, on the set $Y_e$ of the corresponding edge gadget.
		Since $\ori{G}$ is circulating, 
		the latter contribute with an additional $\frac{1}{2}W_v$ colors in total.
		Therefore, we can distribute the remaining $k-\frac{3}{2}W_v-1$ colors 
		on the set $P_v$, which makes $v$ a $b$-vertex.
		This satisfies~\cref{enum:COtoCOL:2},
		and we have arrived at a $b$-coloring of $H$ with $k$ colors.
	\end{claimproof}
	
	We now work towards the reverse implication of the correctness proof.
	We start with a claim regarding the location of the $b$-vertices 
	in any $b$-coloring of $H$ with $k$ colors.
	Throughout the following, we denote by $A$ the set of centers of the anonymous stars.
	\begin{nestedclaim}\label{claim:b-vertices}
	    Each $b$-coloring of $H$ with $k$ colors has precisely one $b$-vertex per color.
	    Moreover, the $b$-vertices are $\{s^\star\} \cup V(G) \cup \calQ \cup A$,
	    and for each $e = uv \in E(G)$, 
	    precisely one of $x_{e, u}$ and $x_{e, v}$.
	\end{nestedclaim}
	\begin{claimproof}
	    The only vertices with high enough degree (at least $k-1$) to become $b$-vertices in such a coloring of $H$
	    are in $\{s^\star\} \cup V(G) \cup \calQ \cup A \cup X$.
	    Note that this set has size $2\bfW + 4m + n + 2 = k + m$.
	    
	    We argue that the gadget of each edge $e = uv$ can contain at most three $b$-vertices.
	    Note that only four of its vertices, $x_{e, u}$, $x_{e, v}$, $q_{e, 1}$, and $q_{e, 2}$
	    have high enough degree to be $b$-vertices.
	    Suppose for a contradiction that $x_{e, u}$ and $x_{e, v}$ are $b$-vertices for colors 
	    $c_u$ and $c_v$, respectively, where $c_u \neq c_v$.
	    For $x_{e, u}$ to be a $b$-vertex of color $c_u$, 
	    it needs to have a neighbor colored $c_v$.
	    By the structure of $H$, this vertex has to be contained in $L_{e, u}$.
	    Similarly, we can conclude that $L_{e, v}$ contains a vertex colored $c_u$.
	    But this means that both $q_{e, 1}$ and $q_{e, 2}$
	    have two neighbors colored $c_u$ and two neighbors colored $c_v$.
	    Since $\deg_H(q_{e, h}) = 2\weight(e) + \card{Z_e} + 3 = k$ for all $h \in [2]$,
	    this means that each of these vertices sees at most $k - 2$ colors in its neighborhood,
	    so neither of them is a $b$-vertex.
	    Therefore we can assume from now on that $x_{e, u}$ and $x_{e, v}$ receive the same color.
	    
	    Since we only have $k + m$ vertices of high enough degree to be $b$-vertices,
	    we can only have enough $b$-vertices if each edge gadget has exactly
	    three $b$-vertices, and if all vertices in $\{s^\star\} \cup V(G) \cup A$ are $b$-vertices.
        Now suppose that for some edge $e = uv \in E(G)$, both $x_{e, u}$ and $x_{e, v}$
	    are $b$-vertices for their color.
	    By the structure of $H$, 
	    this implies that the same colors have to appear on $L_{e, u}$ and $L_{e, v}$.
	    But $s^\star$ needs to be a $b$-vertex, now there are $\weight(e) \ge 1$ colors
	    in its neighborhood that repeat. Since $\deg_H(s^\star) = k - 1$,
	    this is not possible. 
	    This yields the claim.
	\end{claimproof}
	
	Throughout the following, we assume that we have a $b$-coloring of $H$ with $k$ colors.
	Again, for each $e \in E(G)$ and $v \in e$,
	we denote by $C_{e, v}$ the set of colors appearing on the vertices $L_{e, v}$.
	We prove an auxiliary claim.
	Note that the second part was already observed in the proof of \cref{claim:b-vertices},
	but we restate and argue it here for easy reference.

	\begin{nestedclaim}\label{claim:aux}
	    ~
	    \begin{enumerate}
	        \item\label{claim:aux:disjoint}
	        For each $e, e' \in E(G)$ and $v \in e$, $v' \in e'$,
	            if $(e, v) \neq (e', v')$, then $C_{e, v} \cap C_{e', v'} = \emptyset$.
	            
	        \item\label{claim:aux:orient}
	        For each $e = uv \in E(G)$, either colors $C_{e, u}$ or colors $C_{e, v}$ appear on $Y_e$;
	        the former if $x_{e, v}$ is a $b$-vertex and the latter if $x_{e, u}$ is a $b$-vertex.
	    \end{enumerate}
	\end{nestedclaim}
	\begin{claimproof}
	    \cref{claim:aux:disjoint}.
	    By \cref{claim:b-vertices}, we know $s^\star$ is a $b$-vertex. Since its degree is $k-1$,
	    all its neighbors must receive distinct colors.
	    Hence \cref{claim:aux:disjoint} follows.  
	    
	    \cref{claim:aux:orient}. By \cref{claim:b-vertices}, either $x_{e, v}$ or $x_{e, u}$ is a $b$-vertex for its color.
	    Suppose that $x_{e, v}$ is a $b$-vertex (the other case is analogous).
	    For $x_{e, v}$ to be a $b$-vertex, the colors $C_{e, u}$ have to appear
	    in its neighborhood. 
	    By \cref{claim:aux}\cref{claim:aux:disjoint}, we have that $C_{e,u}\cap C_{e,v}=\emptyset$. 
	    By \cref{claim:b-vertices},
	    $q_{e,1}$ is a $b$-vertex for its color. Since the degree of $q_{e,1}$ is $k-1$ and $Z_e\cup L_{e,u}\subset N(q_{e,1})$, we have that no color of $C_{e,u}$ appears in $Z_e$. Hence, the colors of $C_{e,u}$ must appear in $Y_e$.
	\end{claimproof}

	We now construct an orientation $\ori{G}$ of $G$.
	For each edge $e = uv \in E(G)$, if $x_{e, u}$ is a $b$-vertex, 
	then we orient $e$ towards $u$, 
	and if $x_{e, v}$ is a $b$-vertex, we orient $e$ towards $v$.
	Note that by \cref{claim:b-vertices}, this is well-defined.
	Throughout the following whenever we write ``$uv$'' 
	for an edge in $\ori{G}$, we mean that 
	the edge $uv$ is directed from $u$ to $v$ in $\ori{G}$.
	The next claim completes the correctness proof of the reduction.
	
	\begin{nestedclaim}
	    For each $v \in V(G)$, $\sum_{uv \in E(\ori{G})} \weight(uv) = \frac{1}{2}W_v$.
	\end{nestedclaim}
	\begin{claimproof}
	    We first show that $\sum_{uv \in E(\ori{G})} \weight(uv) \ge \frac{1}{2}W_v$.
	    By \cref{claim:b-vertices}, $v$ is a $b$-vertex.
	    Moreover, $\deg_H(v) = k + \frac{1}{2}W_v$ since 
	    $v$ has $k - \frac{3}{2}W_v - 1$ neighbors in $P_v$,
	    $W_v$ additional neighbors in the edge gadgets, 
	    $W_v$ additional neighbours in $\calL$',
	    and no other neighbors.
	    This means that for $v$ to be a $b$-vertex, 
	    $v$ needs to see at least $\frac{1}{2}W_v$ colors
	    in $\bigcup_{e \in E(G), v \in e} Y_e$.
	    \cref{claim:aux} then implies that 
	    there is a set of edges $\{e_1, \ldots, e_d\}$ incident with $v$ 
	    and with $\sum_{i \in [d]} \weight(e_i) \ge \frac{1}{2}W_v$
	    such that for all $i \in [d]$, $x_{e_i, v}$ is a $b$-vertex.
	    This implies the inequality by our construction of $\ori{G}$.
	    
	    Now we show that $\sum_{uv \in E(\ori{G})} \weight(uv) \le \frac{1}{2}W_v$.
	    Let $\calY = \bigcup_{e \in E(G)} Y_e$,
	    note that $\card{\calY} = \bfW$,
	    and that to make each $v \in V(G)$ a $b$-vertex,
	    $\frac{1}{2}W_v$ colors must appear in $N_H(v) \cap \calY$ 
	    that are not in $N_H(v) \setminus \calY$. 
	    Moreover, for each $e = uv \in E(G)$, 
	    $Y_e$ has colors that appear in $N_H(u) \setminus \calY$ but not in $N_H(v) \setminus \calY$ 
	    or vice versa by \cref{claim:aux}.
	    Since $\bfW = \sum_{e \in E(G)} \weight(e) = \sum_{v \in V(G)} \frac{1}{2} W_v$,
	    we can conclude that if for some $v \in V(G)$, 
	    $\sum_{uv \in E(\ori{G})} \weight(uv) > \frac{1}{2}W_v$,
	    then there is another $v' \in V(G) \setminus \{v\}$ with
	    $\sum_{uv' \in E(\ori{G})} \weight(uv') < \frac{1}{2}W_{v'}$,
	    contradicting the previous paragraph.
	\end{claimproof}
	
	\begin{nestedclaim}
	    Given a path decomposition of~$G$ of width~$w$,
	    one can construct a path decomposition of~$H$
    	of width at most~$w + 5$
    	in polynomial time and logarithmic space.
	\end{nestedclaim}
	\begin{claimproof}
	    Let $\calB$ be a path decomposition of $G$ of width $w$.
	    We add $s^\star$ to all bags of $\calB$.
	    For each vertex $v \in V(G)$, let $B_v \in \calB$ be a bag containing $v$.
	    We insert a sequence of $\card{P_v}$ bags after $B_v$
	    containing $B_v$, and a unique vertex of $P_v$.
	    For each edge $e = uv \in E(G)$, let $B_e$ be a bag in $\calB$ containing $u$ and $v$.
	    We insert a sequence of $\card{Y_e \cup Z_e \cup L_{e, u} \cup L_{e, v}}$ bags after $B_e$
	    containing $B_e$, $x_{e, u}$, $x_{e, v}$, $q_{e, 1}$, $q_{e, 2}$,
	    and a unique vertex of $L_{e, u} \cup L_{e, v} \cup Y_e \cup Z_e$.
	    Finally, we append a sequence of bags forming a width-$1$ path decomposition of the anonymous stars.
	    Note that this gives a path decomposition of $H$
	    and there is no bag to which we added more than five vertices.
	    It is easy to see that these operations can be performed 
	    within the claimed time and space requirements.
	\end{claimproof}
	
	This concludes the proof of the theorem.
\end{proof}

\bibliographystyle{plain}
\bibliography{references}

\end{document}